\documentclass[journal,onecolumn,draftclsnofoot]{IEEEtran}
\IEEEoverridecommandlockouts

\usepackage{cite}
\usepackage{amsmath,amssymb,amsfonts,amsthm}
\usepackage{algorithmic}
\usepackage{graphicx}
\usepackage{textcomp}
\usepackage{xcolor}
\usepackage{url}
\usepackage{booktabs} 
\def\BibTeX{{\rm B\kern-.05em{\sc i\kern-.025em b}\kern-.08em
    T\kern-.1667em\lower.7ex\hbox{E}\kern-.125emX}}

\usepackage{tikz}
\usetikzlibrary{arrows.meta,positioning,calc,decorations.pathreplacing}

\usepackage{braket}

\newcommand{\Ht}[1]{\ensuremath{\mathcal{H}}}
\newcommand{\inner}[2]{\langle #1 | #2 \rangle}

\newtheorem{theorem}{Theorem}

\newtheorem{definition}{Definition}

\begin{document}

\title{Quantum Compression for Distributed Entanglement 
}

\author{
Jan~\O stergaard,
Shashi Raj Pandey,
Christophe Biscio,
Torben Bach Pedersen,
Petar Popovski%
\thanks{
Jan~\O stergaard, Shashi Raj Pandey, and Petar Popovski are with the Department of Electronic Systems, Aalborg University, Aalborg, Denmark (e-mails: \{jo, srp, petarp\}@es.aau.dk).

Christophe Biscio is with the Department of Mathematics, Aalborg University, Aalborg, Denmark (e-mail: christophe@math.aau.dk).

Torben Bach Pedersen is with the Department of Computer Science, Aalborg University, Aalborg, Denmark (e-mail: tbp@cs.aau.dk).

This work was supported, in part, by the Danish National Research
Foundation (DNRF), through the Center CLASSIQUE, grant nr. 187.

This work was partially presented in \cite{Ostergaard2026}.
}
}

\maketitle
\begin{abstract}
We study compression strategies for multipartite entanglement distribution under uncertainty in the partitioning of the quantum state. When the partition is not known at the time of state preparation, we show that a joint design of the resource state and a family of compression schemes can increase the entanglement across partitions under a fixed transmission budget.
We formulate this as a source coding problem and derive non-asymptotic upper and lower bounds on the achievable average entanglement subject to an average coding rate. We furthermore design an efficient method for jointly optimizing states and lossless compression maps by exploiting the inherent symmetry of weighted Dicke states. 
In the bipartite case, we propose practical constructions that closely approach the derived upper bound, and more generally we provide practical constructions for multipartite settings.
\end{abstract}


\section{Introduction}
Quantum data compression aims to reduce the physical resources such as qubits, gates, and memory required to store, transmit, or simulate quantum information, while retaining as much essential information as possible \cite{wilde2017quantum}. 

It might not be immediately obvious how quantum data compression can be achieved. In closed quantum systems, all admissible transformations are unitary, and therefore preserve the Hilbert space dimension as well as the volume of accessible quantum states, which then precludes any naive “compression” in the classical sense \cite{nielsen:2010}. Furthermore, the no-deletion theorem forbids the deletion of quantum information via unitary operations alone \cite{pati2000no_deleting}. 
Indeed, in contrast to classical information, multiple copies of an unknown quantum state cannot, in general, be faithfully represented by a single copy. Redundant quantum information can therefore not be naively discarded without loss \cite{nielsen:2010}.

To achieve quantum compression, one way is to embed the relevant states into a lower-dimensional subspace that can be faithfully represented using fewer qubits \cite{schumacher1995quantum}. This is done by identifying and isolating the subspace that captures the essential support of the states of interest, and then apply a transformation that maps it onto a smaller quantum memory register. One can hereafter perform projective measurements on the quantum system followed by a partial trace to discard the unwanted qubits with negligible loss of information.

Schumacher demonstrated that for an independent and identically distributed (i.i.d.) quantum source, asymptotically lossless compression is achievable at a rate equal to the von Neumann entropy of the source \cite{schumacher1995quantum}. In the asymptotic regime, where the number of source qubits tends to infinity, the von Neumann entropy constitutes a fundamental lower bound on the compression rate \cite{schumacher1995quantum}. Moreover, as noted by \cite{datta:2013}, 
Schumacher's strong converse theorem establishes that if the compression rate falls strictly below this threshold, the fidelity between the original and reconstructed quantum states decays exponentially with the block length. Thus, one cannot directly translate Schumacher's asymptotical method into a  lossy compression scheme that operates below the von Neumann entropy. 
Nevertheless, 
in many non-asymptotic applications, lossy quantum compression is both necessary and beneficial \cite{devetak2002quantum,devetak2005distillation,datta2013quantum}. It not only eliminates statistical redundancy, but also discards parts of the quantum state that are difficult or impossible to recover. This is particularly powerful in the context of quantum many-body systems, where the full Hilbert space grows exponentially with system size, yet only a small subset of states are physically relevant.

It has also recently become common to design quantum variational autoencoders as a means to compress higher-dimensional states into lower dimensional states \cite{morris2023qecautoencoder,garcia2024zqvae,romero2018qvae,jonathan2016qautoencoder}. Indeed, it was shown in \cite{gao2017efficient} that deep neural networks can efficiently represent most physical states.  

In classical source coding, efficient lossy or lossless compression schemes asymptotically remove statistical redundancy and produce outputs that are essentially incompressible and are generally close to being i.i.d. \cite{gray2009bernoulli}.
In contrast, efficient quantum compression schemes encode information into typical subspaces of the source, and the resulting compressed states generally exhibit nontrivial entanglement across the output qubits \cite{schumacher1995quantum,wilde2017quantum}. 
Thus, while classical compression tends to remove statistical correlations, quantum compression generally encodes information into correlated and often entangled degrees of freedom. 

A common task in quantum information processing is entanglement distribution \cite{nielsen:2010}. An entangled state on 
$n$ qubits is typically generated locally and subsequently partitioned into
$n \geq m\geq 2$ non-overlapping subsystems. The subsystems are then distributed among multiple parties. If the assignment of subsystems to parties is known in advance, one can tailor the state to exhibit high entanglement across the corresponding partitions. However, if the partitioning is not known at the time of state preparation, the state must be constructed such that it exhibits high entanglement across all possible partitions.
In the case of bipartitions, i.e., when the system is divided into two subsystems, it is in sometimes possible to achieve maximal entanglement across all such cuts. This leads to the notion of absolutely maximally entangled (AME) states, which are defined by the property that every bipartition yields a maximally mixed reduced density matrix on the smaller subsystem \cite{Helwig2012AME}. For qubit systems, AME states are known to exist only for 
$n=2,3,5,$ and $6$ qubits \cite{Scott2004Multipartite,Huber2017SevenQubitAME}. Therefore, for larger numbers of qubits, or when considering more general multipartitions, AME states do not exist.

If the subsystem partitioning is not known at the time of state generation, it was shown in~\cite{Ostergaard2026}, that 
quantum data compression could be used to concentrate the available entanglement into fewer qubits, and thereby at runtime improve the entanglement density for the selection of qubits that needed to be communicated. 
This gives rise to a problem with no direct classical analogue. The problem also differs fundamentally from standard quantum source coding. In Schumacher compression and its variants, the source is given \emph{a priori}, and the objective is to design an encoder--decoder pair adapted to that source or source ensemble. In contrast, here the source state itself becomes a design variable. Thus, the goal is not merely to compress a given source. Instead, the goals are twofold: \emph{a)} to construct a source, i.e., a quantum state, whose entanglement structure is robust under any partitioning and yet is compressible, and \emph{b)} design a family of distributed compressors that can  compress individual subsystems of the source without destroying the entanglement. 

Compared to \cite{Ostergaard2026}, this paper provides (i) a formal coding-theoretic formulation with joint optimization of states and compression maps, (ii) non-asymptotic entanglement–rate bounds, (iii) a multipartite extension to $m \geq 2$, and (iv) a noisy simulation study. These additions establish both fundamental limits and practical performance of the proposed framework.

\subsection*{Paper organization}
This paper is organized as follows. In Section~\ref{sec:motivation}, we motivate the proposed coding problem further by relating it to a distributed entanglement application. 
In Section~\ref{sec:problem} we introduce and formalize the coding problem. Section~\ref{sec:entanglement} introduces the multipartite entanglement measure that we will be using. Section~\ref{sec:design} contains the design  and optimization of the state and compression maps. Section~\ref{sec:bounds} establishes upper and lower bounds to the coding problem. Section~\ref{sec:simulations} presents simulations studies, and Section~\ref{sec:conclusions} contains the conclusions. Proof of theorems are deferred to the supplementary material.

\section{Problem motivation}\label{sec:motivation}
To motivate the problem further with a practical application, consider the case in which a multipartite quantum state is prepared at a base station and subsequently distributed to multiple parties. The intended use of the state is to provide multipartite entanglement as a resource. This would be a common scenario in quantum networks~\cite{Popovski2025_1Q}, where often the systems are subject to practically motivated constraints such as physical resource constraints and potentially uncertainty in subsystem partitioning.

The reason for having physical resource constraints could be motivated by the fact that the cost of transmitting, storing, or maintaining quantum systems scale with the number of physical qubits. In many architectures, it is significantly easier to generate large multipartite states than it is to distribute or store them across a network. As a result, the number of qubits that can be delivered to each party may be strictly limited, even when substantial entanglement is available at the source.

The assignment of physical qubits to subsystems may be predetermined in certain network architectures. However, one can also consider scenarios in which this assignment is not known at the time of state preparation. Such uncertainty may arise in reconfigurable quantum networks, satellite-based quantum communication systems, and quantum internet architectures.
For example, in quantum internet architectures, entanglement is often treated as a network resource that is distributed throughout a network and later consumed~\cite{Wehner2018,Perseguers2013}. This supports situations where the set of active users and the allocation of physical qubits to network nodes only become known after the underlying multipartite resource state has been prepared.

As a concrete example, consider a satellite-based quantum network in which satellites have only brief communication windows while passing within range of a ground station. The first satellite to establish a link may request as many available qubits as possible during its limited access window, leaving only the remaining qubits for subsequent satellites.  The final allocation of physical qubits to the subsystems is determined adaptively and cannot be assumed to be fixed at the time the global quantum state is prepared.
The prepared state should perform well not only for a single fixed partition, but on average across a family of possible partitions. We illustrate such a situation in Fig.~\ref{fig:network}. The figure illustrates the case of $m=2$ parties. An $n=7$-qubit state $\ket{\phi}$ is initially prepared at the base station. The two parties then request a certain number of qubits; for instance, Alice requests three qubits and Bob requests four. The base station randomly assigns three qubits to Alice and the remaining four to Bob, and subsequently compresses Alice’s subsystem into two qubits and Bob’s subsystem into three qubits before transmitting them to the respective parties. 
Compression is employed to reduce the required communication resources and to increase the amount of entanglement available per transmitted qubit. We consider a version of this example in the simulation section.


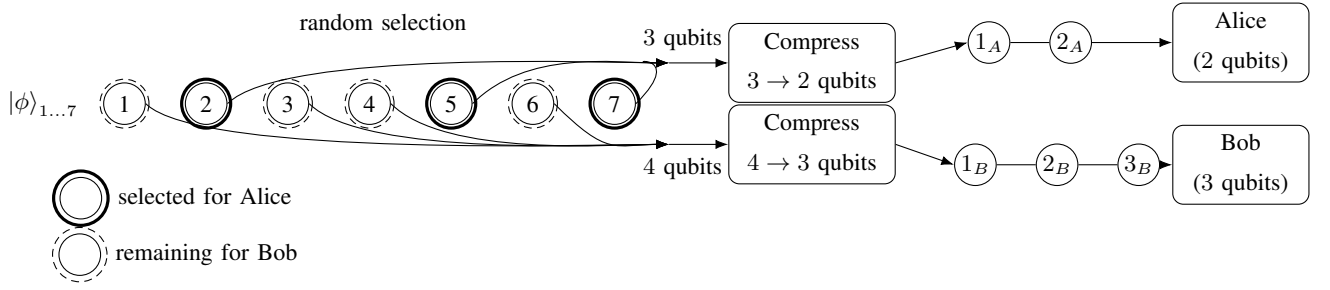
\begin{figure*}
\begin{tikzpicture}[scale=0.9,
    font=\small,
    >=Latex,
    qubit/.style={circle, draw, minimum size=5.5mm, inner sep=0pt},
    box/.style={rectangle, draw, rounded corners, minimum width=22mm, minimum height=10mm, align=center},
    party/.style={rectangle, draw, rounded corners, minimum width=18mm, minimum height=9mm, align=center},
    lab/.style={inner sep=1pt}
]

\node[lab] (phi) at (0,2.0) {$\ket{\phi}_{1\ldots 7}$};

\foreach \i in {1,...,7} {
  \node[qubit] (q\i) at (1.2*\i, 2.0) {\i};
}

\node[lab] (rand) at (5.0, 3.2) {random selection};

\foreach \i in {2,5,7} {
  \draw[very thick] (q\i) circle (3.6mm);
}
\node[lab] (selA) at (9.4, 2.95) {$3$ qubits};

\foreach \i in {1,3,4,6} {
  \draw[densely dashed] (q\i) circle (3.6mm);
}
\node[lab] (selB) at (9.4, 1.05) {$4$ qubits};

\coordinate (splitA) at (9.2,2.6);
\coordinate (splitB) at (9.2,1.4);

\foreach \i in {2,5,7} {
  \draw[->] (q\i.east) .. controls +(0.8,0.8) and +(-0.8,0.0) .. (splitA);
}
\foreach \i in {1,3,4,6} {
  \draw[->] (q\i.east) .. controls +(0.8,-0.8) and +(-0.8,0.0) .. (splitB);
}

\node[box] (compA) at (11.3,2.6) {Compress\\$3 \to 2$ qubits};
\node[box] (compB) at (11.3,1.4) {Compress\\$4 \to 3$ qubits};

\draw[->] (splitA) -- (compA.west);
\draw[->] (splitB) -- (compB.west);

\node[qubit] (a1) at (13.9,2.9) {$1_A$};
\node[qubit] (a2) at (15.1,2.9) {$2_A$};
\draw[->] (compA.east) -- (a1.west);
\draw (a1.east) -- (a2.west);

\node[qubit] (b1) at (13.7,1.1) {$1_B$};
\node[qubit] (b2) at (14.9,1.1) {$2_B$};
\node[qubit] (b3) at (16.1,1.1) {$3_B$};
\draw[->] (compB.east) -- (b1.west);
\draw (b1.east) -- (b2.west);
\draw (b2.east) -- (b3.west);

\node[party] (alice) at (17.6,2.9) {Alice\\($2$ qubits)};
\node[party] (bob)   at (17.6,1.1) {Bob\\($3$ qubits)};

\draw[->] (a2.east) -- (alice.west) node[midway, above] {};
\draw[->] (b3.east) -- (bob.west)   node[midway, above] {};

\node[lab, align=left] at (2.0,0.2) {%
\begin{tabular}{@{}l@{}}
\begin{tikzpicture}[baseline=-0.5ex]
\node[qubit] (x) at (0,0) {};
\draw[very thick] (x) circle (3.6mm);
\end{tikzpicture} selected for Alice \\
\begin{tikzpicture}[baseline=-0.5ex]
\node[qubit] (y) at (0,0) {};
\draw[densely dashed] (y) circle (3.6mm);
\end{tikzpicture} remaining for Bob
\end{tabular}
};

\end{tikzpicture}
\caption{An $n=7$ qubit symmetric state $\ket{\phi}$ is prepared. Three randomly chosen qubits are selected for Alice and the remaining four for Bob.
The subsystems are compressed such that the entanglement is confined to a lower-dimensional support, allowing it to be represented losslessly using fewer qubits.
This reduces communication costs and increases the amount of entanglement per transmitted qubit.
The random selection is independent of $\ket{\phi}$. This scenario illustrates the case of $m=2$ parties. In the general case, there can be $2\leq m\leq n$ parties.}
\label{fig:network}
\end{figure*}

\section{Problem Formulation}\label{sec:problem}
We now formalize the coding problem. In contrast to conventional source coding, the source is not given \emph{a priori}, but also needs to be optimized. The objective is to jointly construct a quantum state and a family of distributed compression schemes such that, for a given average coding rate across all $m$-partitions, the average entanglement across these partitions is maximized. 

Let $\ket{\psi} \in (\mathbb{C}^2)^{\otimes n}$ denote the $n$-qubit resource state (source state). We will here focus on the case, where $\ket{\psi}$ is restricted to be a pure state. 
Let $\Pi_{n,m}$ denote the set of all admissible partitions of the $n$ qubits into $m$ non-overlapping subsets. For a given partition $\pi \in \Pi_{n,m}$, we write:
\begin{align}
\pi &= \{S_{\pi,1}, \dots, S_{\pi,m}\}, 
\quad 
S_{\pi,i} \subseteq \{1,\dots,n\}, 
\quad 
S_{\pi,i} \cap S_{\pi,j} = \varnothing \ \text{for } i \neq j, \\
&\quad 
\bigcup_{i=1}^m S_{\pi,i} = \{1,\dots,n\},
\quad 
|S_{\pi,i}| = c_{\pi,i}, 
\quad 
\sum_{i=1}^m c_{\pi,i} = n.
\end{align}
The system is thus viewed as a distributed quantum source with local Hilbert spaces:
\begin{equation}
    \mathcal{H}_{S_{\pi,i}} \cong (\mathbb{C}^2)^{\otimes c_{\pi,i}}, \qquad i = 1, \dots, m.
\end{equation}

For each partition $\pi$, a distributed compression code consists of local completely positive trace-preserving (CPTP) maps
\begin{align}
    \mathcal{E}_{\pi,i} &: \mathcal{D}(\mathcal{H}_{S_{\pi,i}}) \to \mathcal{D}(\mathcal{H}_{R_{\pi,i}}), \\
    \mathcal{D}_{\pi,i} &: \mathcal{D}(\mathcal{H}_{R_{\pi,i}}) \to \mathcal{D}(\mathcal{H}_{S_{\pi,i}}),
\end{align}
where $R_{\pi,i}$ denotes the register holding the compressed information of subsystem $S_{\pi,i}$ of size $r_{\pi,i}$ qubits,  
$\mathcal{H}_{R_{\pi,i}} \cong (\mathbb{C}^2)^{\otimes r_{\pi,i}}$ and $\mathcal{D}(\mathcal{H})$ denotes the set of density operators on $\mathcal{H}$. 
The global encoder and decoder associated with partition $\pi$ are given by
\begin{equation}
    \mathcal{E}_\pi = \bigotimes_{i=1}^m \mathcal{E}_{\pi,i}, 
    \qquad
    \mathcal{D}_\pi = \bigotimes_{i=1}^m \mathcal{D}_{\pi,i}.
\end{equation}
The total rate for representing partition $\pi$ is defined as:
\begin{equation}
    r(\pi) = \sum_{i=1}^m r_{\pi,i}.
\end{equation}
Thus, $1\leq r_{\pi,i}\leq c_{\pi,i}$, where $c_{\pi,i}$ and $r_{\pi,i}$ denote the number of qubits allocated to the subset before and after compression, respectively.
The average coding rate (after compression) across all admissible partitions is then:
\begin{equation}
    R \triangleq \frac{1}{|\Pi_{n,m}|} \sum_{\pi \in \Pi_{n,m}} r(\pi).
\end{equation}

To avoid that the compression task destroys the entanglement, we quantify entanglement based on the \emph{reconstructed} state after compression and decompression. For each partition $\pi \in \Pi_{n,m}$, define
\begin{equation}
    \rho_\pi \triangleq \mathcal{D}_\pi \circ \mathcal{E}_\pi(\ket{\psi}\bra{\psi}).
\end{equation}
We then define the entanglement across partition $\pi$ as
\begin{equation}
    E(\pi) \triangleq E(\rho_\pi),
\end{equation}
where $E(\cdot)$ denotes a chosen entanglement measure across the subsystems specified by $\pi$.
The average entanglement is thus given by:
\begin{equation}
    E \triangleq \frac{1}{|\Pi_{n,m}|} \sum_{\pi \in \Pi_{n,m}} E(\rho_\pi).
\end{equation}

The coding problem can then be formulated as the following joint optimization:
\begin{align} \label{eq:opt_problem}
    \max_{\ket{\psi}, \{\mathcal{E}_{\pi,i}, \mathcal{D}_{\pi,i}\}} 
    \quad & E \\
    \text{subject to} \quad 
    & R \leq R_0, 
\end{align}
where $R_0$ is a prescribed average rate constraint. 
In some applications it might be relevant to replace the average rate constraint by individual rate constraints at each of the subsystems. 
We demonstrate in Section~\ref{sec:bounds} how individual rate constraints can be accommodated naturally within our optimization framework by restricting the feasible set to partitions of symmetric states that satisfy the prescribed rate constraints.

For the optimization problem in \eqref{eq:opt_problem}, then for each partition, $m$ encoders and $m$ decoders need to be specified. The number $|\Pi_{n,m}|$ of partitions is exponential in $n$ for a fixed $m$ and is given by the Stirling number of the second kind \cite{Stanley1999EC2}:
\begin{equation}
    |\Pi_{n,m}| = S(n,m),
\end{equation}
which counts the number of ways to partition an $n$-element set into $m$ non-empty and non-overlapping subsets. This gives rise to an exponential number of encoders-decoders making the general problem intractable. We will therefore later restrict attention to so-called permutation-invariant states and universal lossless symmetric-subspace compressors and show that this makes the problem tractable.

\section{Multipartite Entanglement}\label{sec:entanglement}
In this section, we introduce the notion of entanglement that we will be using.

Let $\rho_{AB}$ be a bipartite $n$-qubit system comprising the two subsystems $A$ and $B$. If $\rho_{AB}$ is a pure state, 
the entanglement between its subsystems can be quantified by the von Neumann entropy $S(\rho_A)=S(\rho_B)$ of the reduced states $\rho_A=\mathrm{Tr}_B(\rho_{AB})$  and $\rho_B=\mathrm{Tr}_A(\rho_{AB})$ \cite{nielsen:2010}:
\begin{align}\label{eq:ent}
 S(\rho_A) = -\mathrm{Tr}[\rho_A \log \rho_A].
\end{align}

For multipartite quantum systems, entanglement cannot in general be fully
characterized by bipartite entropy measures. For example, in a tripartite system
described by a state $\rho_{ABC}$, the von Neumann entropy of a reduced state,
such as $S(\rho_A)$, describes the \emph{bipartite} entanglement over the cut $A|BC$ but does 
not provide a complete description of the entanglement
shared among the subsystems $A$, $B$, and $C$. 

Several complementary approaches to quantifying multipartite entanglement have
been proposed in the literature. 
For example, there are distance-based measures, such as the relative entropy of entanglement, which offer a
conceptually appealing information-theoretic characterization but which are
computationally intractable for large multipartite systems~\cite{Vedral1997}.
Other measures include the three-tangle measure, which is a polynomial
entanglement invariant that quantifies genuine tripartite entanglement in
three-qubit systems and, in particular, detects GHZ-type correlations~\cite{CKW2000}.
However, the three-tangle measure does not admit a straightforward generalization to
larger multipartite systems. Moreover, the three-tangle entanglement vanishes for the
three-qubit W state. This indicates an absence of GHZ-type 
tripartite entanglement, even though the W state exhibits a high degree of other forms of
multipartite entanglement such as that captured by the geometric measure of entanglement (GME)\cite{WeiGoldbart2003}.

In its standard form, the GME is defined relative to the set of
fully separable states and thus probes genuine $n$-partite entanglement.
More generally, one may define a \emph{partition-based} or
\emph{block-product} geometric entanglement by restricting the optimization
to states that are separable with respect to a fixed multipartition of the
system~\cite{WeiGoldbart2003,Huber2010,Blasone2008}.
Specifically, let $\ket{\psi}\in(\mathbb{C}^2)^{\otimes n}$ be an $n$-qubit pure
state, and let $\lambda=\{A_1,\dots,A_m\}$ be a fixed partition of the
qubit indices $\{1,\dots,n\}$ into $m$ disjoint, non-empty blocks (an $m$-cut),
with $A_j\cap A_k=\emptyset$ for $j\neq k$ and
$\bigcup_{j=1}^m A_j=\{1,\dots,n\}$.  The set of $\lambda$-separable pure
states is defined as:
\begin{equation}
\mathrm{Sep}(\lambda)
=
\left\{
\ket{\phi}:\;
\ket{\phi}=\bigotimes_{j=1}^m \ket{\phi_{A_j}},
\quad
\ket{\phi_{A_j}}\in(\mathbb{C}^2)^{\otimes |A_j|}
\right\}.
\end{equation}

The geometric entanglement of $\ket{\psi}$ with respect to the partition
$\lambda$ is then given by
\begin{align}
E_G^{\lambda}(\ket{\psi})
&=
-\log_2 \Lambda_{\lambda}^2(\ket{\psi}), \\
\Lambda_{\lambda}^2(\ket{\psi})
&=
\max_{\ket{\phi}\in \mathrm{Sep}(\lambda)}
\left|\inner{\phi}{\psi}\right|^2.
\label{eq:EG_mpartition}
\end{align}

By construction, $E_G^{\lambda}(\ket{\psi})=0$ if and only if
$\ket{\psi}$ is a product state across the blocks of $\lambda$.

Note that, for a permutation-invariant state $\ket{\psi}$, all partitions sharing the same block-size multiset $\lambda$ yield the same geometric entanglement. Hence, we may equivalently optimize over $\lambda \in \Lambda_{n,m}$ with multiplicities $f(\lambda)$, and define
$E_G(\lambda) \triangleq E_G^{\lambda}(\ket{\psi})$.

\section{Joint Design of State and Compression Maps}\label{sec:design}
In this section, we introduce the type of quantum states that we will focus on. We furthermore present the compression schemes, and then rewrite the general coding problem \eqref{eq:opt_problem} into a simpler form, which is practically solvable under certain structural assumptions. 

\subsection{Permutation-Invariant States}
Due to the exponential number of distinct partitions, it quickly becomes infeasible to design individual compression schemes for each partition for the problem in 
\eqref{eq:opt_problem}. To reduce the complexity, it is advantageous to use an efficient representation of the states and compressors. One such representation is given by permutation invariant states, which are compressible under arbitrary partitions.

The Dicke states $\ket{D_k^{(n)}}$ of $n$ qubits with $k$ excitations are permutation-invariant states defined as \cite{Dicke1954,Marconi2025symmetric}:
\begin{align}
\ket{D_k^{(n)}} = \frac{1}{\sqrt{\binom{n}{k}}} \sum_{\substack{\mathrm{wt}(x) = k \\ x \in \{0,1\}^n}} \ket{x},
\end{align}
where the sum is over all computational basis states $\ket{x}$ of Hamming weight \( \mathrm{wt}(x)= k\), and \(\binom{n}{k}\) is the binomial coefficient. Choosing $k=\lfloor n/2 \rfloor$ and making an equal split of the system into two parties maximizes their entanglement over all choices of $k$ \cite{Stockton2003}. 

Another kind of symmetric (permutation invariant) states are the Comb states, which refers to a family of generalized Dicke states whose Hamming weights are evenly spaced by a step size $s$ \cite{Stockton2003}:
\begin{align}\label{eq:comb}
\ket{C(s)}
=
\frac{1}{\sqrt{2L+1}}
\sum_{\ell=-L}^{L} \ket{D^{(n)}_{n/2+\ell s}},
\end{align}
where $L$ is the largest integer such that
$0 \le n/2 + \ell s \le n$ for all $\ell \in \{-L,\dots,L\}$. For large $n$ and $s \approx \sqrt{2n}$, the number of terms satisfies $2L+1 \approx n/s$, which asymptotically maximizes the entanglement~\cite{Stockton2003}.

Since Dicke states forms an orthonormal basis in the symmetric subspace, any linear combination of them is automatically permutation invariant, cf. \cite{Keyl2002, Christandl2007}.
Specifically, let $\mathcal{K}=\{k_1,\dotsc, k_K\}$ be the set of $K$ distinct Hamming weights, each satisfying $0\leq k_i\leq n$. Moreover, let  $\mathcal{A} = \{\alpha_1,\dotsc, \alpha_K\}$, where $\alpha_i \in \mathbb{C}$ are the complex amplitudes used to weight the individual Dicke states. Then, a \emph{weighted} superposition of Dicke basis states with different Hamming weights can be constructed as follows:
\begin{align}\label{eq:mhw}
\ket{\psi_{\mathcal{K},\mathcal{A}}^{(n)}} = \frac{1}{\sqrt{\sum_{i=1}^K |\alpha_i|^2}} \sum_{i =1}^K\alpha_i 
\ket{D_{k_i}^{(n)}},
\end{align}
where $\alpha_i$ describes the complex amplitude (weight) for the $i$-th Dicke state. It can be shown that arbitrarily selected subsystems of $\ket{\psi_{\mathcal{K},\mathcal{A}}^{(n)}}$ are also permutation invariant \cite{Keyl2002,Christandl2007,Ostergaard2026}.

\subsection{Universal Compression Maps for Symmetric States}
The support of permutation invariant states lies entirely within the symmetric subspace. For an $n$-qubit system, the symmetric subspace $\mathcal H_{\mathrm{sym}}^{(n)} \subset (\mathbb{C}^2)^{\otimes n}$ has dimension $n+1$~\cite{plesch2010efficient,Harrow2013ChurchSymmetricSubspace,PivoluskaPlesch2022}. Therefore, any state supported on $\mathcal H_{\mathrm{sym}}^{(n)}$ can be represented using only $\lceil \log_2(n+1) \rceil$ qubits via an isometric embedding.

Now consider an $n$-qubit state that is partitioned into $m$ subsystems, where each subsystem is permutation invariant. Let the $i$th subsystem consist of $c_i$ qubits, with $\sum_{i=1}^m c_i = n$. Each subsystem can then be independently compressed into a register of size $r_i$ qubits, where
\begin{equation}
    r_i = \left\lceil \log_2(c_i + 1) \right\rceil.
\end{equation}
The total amount $r$ of qubits after compression is:
\begin{equation}
    r = \sum_{i=1}^m r_i.
\end{equation}

We now define the corresponding encoding and decoding maps. Let $\{\ket{D_w^{(c_i)}}\}_{w=0}^{c_i}$ denote the Dicke basis spanning $\mathcal H_{\mathrm{sym}}^{(c_i)}$. 
Let the isometry $V_i$ be defined as:
\begin{equation}
    V_i : \mathcal H_{\mathrm{sym}}^{(c_i)} \to (\mathbb{C}^2)^{\otimes r_i},
\end{equation}
where
\begin{equation}
    V_i \ket{D_w^{(c_i)}} = \ket{w}_{R_i}, \quad w = 0, \dots, c_i.
\end{equation}
Thus, $V_i$ can be seen as a compression map  that replaces a symmetric $c_i$-qubit state by the integer $w$ (its excitation number). While the original state requires $c_i$ qubits, the compressed state only requires $\lceil \log_2(c_i+1)\rceil$ qubits.

The map $V_i$ furthermore satisfies:
\begin{equation}
    V_i^\dagger V_i = \Pi_{\mathrm{sym}}^{(c_i)},
\end{equation}
where $\Pi_{\mathrm{sym}}^{(c_i)}$ is the projector onto the symmetric subspace. Thus, $V_i$ is invertible on the symmetric subspace. 
The local encoding and decoding maps are defined as
\begin{equation}
    \mathcal E_i(\rho) = V_i \rho V_i^\dagger, 
    \qquad
    \mathcal D_i(\sigma) = V_i^\dagger \sigma V_i.
\end{equation}
Strictly speaking, $\mathcal E_i$ and $\mathcal D_i$ are trace-preserving only on states supported on $\mathcal H_{\mathrm{sym}}^{(c_i)}$, since $V_i^\dagger V_i = \Pi_{\mathrm{sym}}^{(c_i)}$. In the following, we restrict attention to this subspace, on which the maps act as isometries.
By construction, for any state $\rho$ supported on $\mathcal H_{\mathrm{sym}}^{(c_i)}$, we have
\begin{equation}
    \mathcal D_i \circ \mathcal E_i(\rho) = \rho,
\end{equation}
and thus the compression is lossless, when the input state is symmetric.
The global compression map is given by
\begin{equation}
    V = \bigotimes_{i=1}^m V_i,
\end{equation}
which acts isometrically on states supported on $\bigotimes_{i=1}^m \mathcal H_{\mathrm{sym}}^{(c_i)}$. For such states $\ket{\psi}$, the compressed representation
$ \ket{\psi_{\mathrm{comp}}} = V \ket{\psi}$
preserves all quantum information, including entanglement.

The isometry $V_i$ is linear and depends only on the subsystem sizes $c_i$ and not on the specific input state. Therefore, the same encoder--decoder pair applies uniformly to all states supported on the symmetric subspaces. In this sense, the construction provides a universal, lossless compression scheme for permutation-invariant states.

\subsection{Access to entanglement in compressed domain}
\label{subsec:structural_access}
We now show that, for lossless compression into symmetric subspaces, the entanglement is faithfully represented in the compressed code space and can be accessed there directly. Hence, decompression is unnecessary.

Let $\pi=\{S_1,\dots,S_m\}$ be a partition with $|S_i|=c_i$, and assume that each subsystem is supported on the symmetric subspace
$
\mathcal H_{\mathrm{sym}}^{(c_i)} \subset (\mathbb C^2)^{\otimes c_i}.
$
For each block define a local compression isometry:
\begin{equation}  
V_i:\mathcal H_{\mathrm{sym}}^{(c_i)} \to \mathcal H_{R_i}\cong (\mathbb C^2)^{\otimes r_i},
\quad
r_i=\lceil \log_2(c_i+1)\rceil,
\end{equation}
and let $V_\pi=\bigotimes_{i=1}^m V_i$.
Then any physical state $\ket{\psi}$ supported on $\bigotimes_i \mathcal H_{\mathrm{sym}}^{(c_i)}$ is mapped to the compressed state
\begin{equation}
\ket{\psi_{\mathrm{comp}}}=V_\pi\ket{\psi}\in \bigotimes_{i=1}^m \mathcal H_{R_i}.
\end{equation}

Since $V_\pi$ is a product of local isometries, it preserves all entanglement properties with respect to the partition $\pi$. In particular, the compressed registers $R_1,\dots,R_m$ already form a valid multipartite entanglement resource. The relevant subsystem structure after compression is therefore the logical tensor product
$
\mathcal H_{\mathrm{code}}^{(\pi)}=\bigotimes_{i=1}^m \mathcal H_{R_i},
$
rather than the original qubit-wise factorization within each block.
Moreover, any operator $O_i$ acting on the symmetric support of subsystem $S_i$ has a logical representation
$
\widetilde O_i = V_i O_i V_i^\dagger
$
on $R_i$, so that
\[
\bra{\psi}\bigotimes_{i=1}^m O_i \ket{\psi}
=
\bra{\psi_{\mathrm{comp}}}\bigotimes_{i=1}^m \widetilde O_i \ket{\psi_{\mathrm{comp}}}.
\]
Hence measurements, local processing, and entanglement-assisted protocols that are defined on the support of the symmetric subspace can be implemented directly in the compressed domain. Thus, no decompression is required when using this lossless compression map. Decompression is only needed if one wishes to recover the original physical qubit representation as would be the natural objective in classical data compression.

\subsection{Computation of Entanglement}
It can  be shown that for permutation-invariant states, the GME only depends upon the sizes of the cuts and not the position of the individual qubits \cite{Hubener2009}. 
Thus, an $m$-cut is fully characterized by the multiset of block sizes, and the number of inequivalent $m$-cuts is therefore given by the number $p_m(n)$ of integer partitions of $n$ into exactly $m$ positive parts, which satisfies the recurrence relation
\begin{align}
p_m(n) = p_m(n-m) + p_{m-1}(n-1),
\end{align}
with the conditions $p_0(0)=1, p_m(n)=0$ for $m>n$, and $p_0(n) = 0$ for all $n$. 

Let $\ket{\psi}$ be an $n$-qubit permutation-invariant pure state, and let
$\lambda = \{\lambda_1,\dots,\lambda_m\}$ be an unordered multiset of positive
integers satisfying $\sum_{i=1}^m \lambda_i = n$.
Then the GME  associated with the partition
$\lambda$ is well defined and is independent of the particular choice or ordering of the subsystems
realizing the partition.
Thus, for any two $m$-cut partitions $\lambda$ and $\lambda'$ with the
same block-size multiset, $E_G^{\lambda}(\ket{\psi}) = E_G^{\lambda'}(\ket{\psi})$.

Let $\Lambda_{n,m}$ denote the set of all nontrivial block-size multisets:
\begin{align}\notag
\Lambda_{n,m}
=
\Bigl\{
&\lambda = (\lambda_1,\dots,\lambda_m)\in\mathbb{N}^m :
1 \le \lambda_1 \le \cdots \le \lambda_m, \\ \label{eq:Lambdaset}
&\sum_{i=1}^m \lambda_i = n
\Bigr\}, \quad |\Lambda_{n,m}| = p_m(n).
\end{align}

As an example, assume that $n=6$ and $m=3$. Then, there are three distinct partitions, i.e., $\Lambda_{6,3}=\{(1,1,4)$, $(1,2,3), (2,2,2)\}$. This greatly reduces the computational complexity since the number $p_m(n)$ of distinct cuts to consider is polynomial in $n$ rather than exponential. 
Of course, the different partitions do not necessarily occur an equal amount of times. Thus, to compute the average entanglement and coding rate one needs to include the frequency weights $f(\lambda)$, that is:
\begin{align}\label{eq:eer2_correct}
E
&=
\frac{1}{S(n,m)}
\sum_{\lambda\in \Lambda_{n,m}}
f(\lambda)\,
E_G^{\lambda}(\ket{\psi}), \\
R&= \frac{1}{S(n,m)}
\sum_{\lambda\in \Lambda_{n,m}}
f(\lambda)\,\sum_{i=1}^m r_{\lambda,i},
\end{align}
where $f(\lambda)$ denotes the number of distinct and non-trivial partitions that are a permutation of $\lambda$. 
For $\lambda=(\lambda_1,\dots,\lambda_m)\in\Lambda_{n,m}$, let $s_1,\dots,s_t$
be the distinct values appearing in $\lambda$ with multiplicities
$\mu_1,\dots,\mu_t$ (so $\sum_{j=1}^t \mu_j=m$). The number of set partitions
of $\{1,\dots,n\}$ into $m$ unlabeled blocks with block-size multiset $\lambda$
is then \cite{Stanley1999EC2}:
\begin{align}\label{eq:lambda}
f(\lambda)
&=
\frac{n!}{\prod_{i=1}^m \lambda_i!}\;
\frac{1}{\prod_{j=1}^t \mu_j!} \\
\sum_{\lambda \in \Lambda_{n,m}} f(\lambda) &= S(n,m).
\end{align}
For example, for $n=6$ and $m=3$, let $\lambda = (1,1,4), \lambda' = (1,2,3)$, and $\lambda'' = (2,2,2)$, then 
$f(\lambda)=15$ and $f(\lambda')=60,$ and $f(\lambda'')=15$.

\subsection{Optimization Algorithm}
Under the symmetry restriction, the $n$-qubit state is fully described by a small set of amplitudes over Dicke weights. This reduces the original infinite-dimensional optimization over $(\mathbb{C}^2)^{\otimes n}$ to a finite-dimensional problem over a low-dimensional space. We now describe how this structure can be exploited to efficiently obtain approximate solutions to the coding problem.

Assume the number of qubits $n$ and the number of parties (blocks) $m$ is known.
Then, we are interested in finding a good permutation-invariant state, which is supported on a small set of Dicke weights:
$\mathcal{K}=\{k_1,\ldots,k_K\}\subseteq\{0,1,\ldots,n\}$ with $|\mathcal{K}|=K$. For a given choice of $\mathcal{K}$ and complex amplitudes $\boldsymbol{\alpha}=(\alpha_1,\ldots,\alpha_K)\in\mathbb{C}^K$ satisfying $\|\boldsymbol{\alpha}\|_2^2=\sum_{j=1}^K|\alpha_j|^2=1$, the corresponding $n$-qubit resource state is
\begin{equation}\label{eq:optprob}
\ket{\psi_{\mathcal{K},\boldsymbol{\alpha}}^{(n)}}
=\sum_{j=1}^{K}\alpha_j\,\ket{D_{k_j}^{(n)}},
\qquad
\sum_{j=1}^{K}|\alpha_j|^2=1,
\end{equation}
where $\ket{D_{k_j}^{(n)}}$ denotes the (normalized) Dicke state of Hamming weight $k_j$.

Under the symmetry restriction and the Dicke-state parametrization introduced above, the joint optimization problem in \eqref{eq:opt_problem} reduces to selecting a finite set of Hamming weights and their associated amplitudes. This yields a structured, finite-dimensional optimization problem, which we formulate as follows:
\begin{align}\label{eq:final_opt}
\max_{\substack{\mathcal{K}\subseteq\{0,\ldots,n\}\\ |\mathcal{K}|=K}}
\ \max_{\substack{\boldsymbol{\alpha}\in\mathbb{C}^K\\ \|\boldsymbol{\alpha}\|^2_2=1}}
\ E \\
\mathrm{subject\ to:} & \ R \leq R_0.
\end{align}
Equation~\eqref{eq:final_opt} is a mixed discrete--continuous, highly non-convex optimization problem: the outer maximization selects the set of Hamming weights $\mathcal{K}$, while the inner maximization optimizes the complex amplitudes $\boldsymbol{\alpha}$ on the unit sphere.

For a fixed Hamming-weight support $\mathcal K$, the amplitudes
$\boldsymbol{\alpha}\in\mathbb C^K$ are optimized numerically using
the Sequential Least Squares Quadratic Programming (SLSQP) algorithm.
In the implementation, the complex vector is represented as a
real vector $x = (\operatorname{Re}\alpha_1,\ldots,\operatorname{Re}\alpha_K,
     \operatorname{Im}\alpha_1,\ldots,\operatorname{Im}\alpha_K)$,
and the normalization constraint is imposed as $
\sum_{k=1}^K \left(\operatorname{Re}(\alpha_k)^2
+\operatorname{Im}(\alpha_k)^2\right)=1$.    
SLSQP minimizes the negative average geometric entanglement, so that the
corresponding maximization problem is:
$
\max_{\boldsymbol{\alpha}}
\ \frac{1}{|\Lambda_{n,m}|}
\sum_{\lambda\in\Lambda_{n,m}}
E_G^\lambda(\ket{\psi(\boldsymbol{\alpha})})$
subject to $\|\boldsymbol{\alpha}\|_2^2 = 1$.
For each evaluation of the objective function, the corresponding weighted Dicke state is constructed and the average GME is computed using the alternating product-state optimization described below.

To compute the GME across an $m$-partition $\lambda=\{S_1,\ldots,S_m\}$ of $\{1,\ldots,n\}$, we consider the following optimization problem:
\begin{equation}
E_G^{\lambda}(\ket{\psi})
:=
-\log_2\!\left(
\max_{\ket{\phi_i}\in\mathcal{H}_{S_i}}
\left|\bra{\phi_1\otimes\cdots\otimes\phi_m}\ket{\psi}\right|^2
\right),
\end{equation}
which tries to find the maximal overlap between $\ket{\psi}$ and a fully product state compatible with the partition. Thus, we need to find a state $\ket{\psi}$ that solves:
\begin{align}
\Lambda_{\max}^\lambda(\ket{\psi})
&\triangleq
\max_{\substack{\|\phi_i\|=1\\ i=1,\ldots,m}}
\left|\bra{\phi_1\otimes\cdots\otimes\phi_m}\ket{\psi}\right|.
\end{align}

Since this optimization is nonconvex and admits no known closed-form solution in general, we approximate $\Lambda_{\max}^\lambda$ numerically using an alternating optimization procedure, which is standard in the study of geometric entanglement \cite{BaiYanZeng2023}. The method proceeds as follows. Fix all local states except $\ket{\phi_i}$, and optimize the overlap with respect to $\ket{\phi_i}$ alone. For a fixed $\{\ket{\phi_j}\}_{j\neq i}$, the optimal update is obtained by projecting $\ket{\psi}$ onto the tensor product of the remaining local states,
\begin{equation}
\ket{v_i}
=
\left(
\bra{\phi_1}\otimes\cdots\otimes\bra{\phi_{i-1}}
\otimes I_{S_i}
\otimes\bra{\phi_{i+1}}\otimes\cdots\otimes\bra{\phi_m}
\right)\ket{\psi}.
\end{equation}
The update is then obtained by normalization: $\ket{\phi_i} = \ket{v_i}/ \|\ket{v_i}\|$. Cycling through $i=1,\ldots,m$ defines one sweep of the algorithm.

Starting from a random initial product state, this update rule is iterated until convergence of the overlap. To mitigate convergence to suboptimal local maxima, we perform multiple random restarts and keep the largest overlap observed. 
Since the alternating optimization is not guaranteed to find the global optimum, the obtained overlap $\Lambda$
 provides a lower bound on $\Lambda_\mathrm{max}$, and hence the corresponding value of $E_G(\ket{\psi};\pi)$  constitutes an upper bound on the GME.

As previously mentioned, we do not need to check all $m$-partitions. It suffices to consider only inequivalent partitions classified by their block-size profiles, i.e., those $m$-partitions that are in $\Lambda_{n,m}$ \eqref{eq:Lambdaset} \cite{Hubener2009,MarkhamVirmani2010}.

\section{Entanglement--Rate Bounds}\label{sec:bounds}
In this section, we characterize the fundamental entanglement--rate trade-off associated with the coding problem in \eqref{eq:final_opt}. Specifically, we derive upper bounds that limit the maximum achievable average entanglement under a given rate constraint. Moreover, we derive lower bounds obtained via explicit constructions of permutation-invariant states and lossless compression schemes. We will use these bounds in the simulation section, when evaluating the performance of the proposed optimization algorithm. 

While our results are derived under an average rate constraint across all possible partitions, they can be generalized straightforwardly to individual subsystem rate constraints. For example, consider $n=7$ qubits and suppose that subsystem 1 is subject to a rate constraint $R_1 =2$  while subsystem 2 is subject to a rate constraint $R_2=3$. In this case, 
the feasible set must be restricted to partitions that satisfy both constraints. If $k$ qubits are given to subsystem 1 and $n-k$ to subsystem 2, the corresponding coding rates are $\lceil \log_2(k+1)\rceil$ and $\lceil \log_2(n-k+1)\rceil)$, respectively. This implies that the admissible partitions are  $(1,6), (2,5)$, and $(3,4)$, whereas $(4,3), (5,2)$, and $(6,1)$ are excluded since $\lceil \log_2(k+1)\rceil >R_1$  violate the rate constraint associated with subsystem 1 for $k\geq 4$.

\subsection{Two parties: $m=2$}
For the case of $m=2$, and restricting attention to lossless coding schemes
with pure reconstructed states, we can use the von Neumann entropy as a measure
of bipartite entanglement.
\begin{definition}[Average von Neumann entanglement]
For any $n\geq 2$, we define $\overline E_{\mathrm{vN}}$ as the average von Neumann entanglement:
\begin{align}
\overline E_{\mathrm{vN}}
&\triangleq
\frac{1}{S(n,2)}
\sum_{\lambda \in \Lambda_{n,2}}
f(\lambda)\, S(\rho_\lambda),
\end{align}
where $S(\rho_\lambda)$ denotes the von Neumann entropy of any of the reduced states
of the reconstructed pure state corresponding to any bipartition with block-size
pattern $\lambda=(k,n-k)$.
\end{definition}

\begin{theorem}[Upper bound for $m=2$]\label{thm:tradeoff_m2}
Let $\ket{\psi}$ be an $n$-qubit permutation-invariant pure state, and consider all bipartitions $\pi = A \mid A^c$. Let $k = \min(|A|, |A^c|)$ denote the size of the smaller subsystem, and let $r_k$ denote the total compression rate assigned to all bipartitions with smaller subsystem size $k$, i.e., $r_k = r_A + r_{A^c}$.
Then the average rate is
\begin{equation}
R =
\frac{1}{2^{n-1}-1} \sum_{k=1}^{\lfloor n/2 \rfloor} f(k,n-k) \, r_k,
\end{equation}
and every admissible lossless coding scheme with pure reconstructed state satisfies
\begin{equation}
\overline{E}_{\mathrm{vN}}
\le
\frac{1}{2^{n-1}-1}
\sum_{k=1}^{\lfloor n/2 \rfloor}
f(k,n-k) \,
\min\!\left\{
\left\lfloor \frac{r_k}{2} \right\rfloor,\,
\log_2(k+1)
\right\}.
\end{equation}
\end{theorem}

\begin{theorem}[Lower bound for $m=2$]\label{thm:low_m2}
If the average rate satisfies:
\begin{equation}
R\geq 
\frac{1}{2^{n-1}-1}
\sum_{k=1}^{\lfloor n/2\rfloor}
f(k,n-k)
\left(
\left\lceil \log_2(k+1)\right\rceil
+
\left\lceil \log_2(n-k+1)\right\rceil
\right),
\end{equation}
Then, the average entanglement  is lower bounded by:
\begin{equation}
\overline E_{\mathrm{vN}}
\ge
\frac{1}{2^{n-1}-1}
\sum_{k=1}^{\lfloor n/2\rfloor}
f(k,n-k)
\left(
-\sum_{j=0}^{k}\lambda_j^{(k)}\log_2\lambda_j^{(k)}
\right),
\end{equation}
where
\begin{equation}
\lambda_j^{(k)}
=
\frac{\binom{k}{j}\binom{n-k}{t-j}}{\binom{n}{t}},
\qquad
t=\lfloor n/2\rfloor .
\end{equation}
\end{theorem}

\subsection{Multipartite case: $m\geq 3$}
For the case of $m\geq 3$, we can use the GME as a measure for entanglement between the $m$-partitions. 

\begin{theorem}[Upper bound for $m\geq 3$]\label{thm:ubm}
Let $\ket{\psi}$ be an $n$-qubit permutation-invariant pure state and let
$\pi=\{S_1,\dots,S_m\}$ be an $m$-partition with block sizes
\[
|S_i|=\lambda_i,\quad
1\le \lambda_1\le \cdots \le \lambda_m,\quad
\sum_{i=1}^m \lambda_i=n.
\]
Suppose the local encoders map the $m$ subsystems into registers of sizes
$r_1,\dots,r_m$ qubits, respectively, so that the total rate for partition $\pi$ is
\begin{equation}
r(\pi)=r_1+\cdots+r_m \triangleq r_\lambda.
\end{equation}
Then, the average rate is given by:
\begin{equation}
R
=
\frac{1}{S(n,m)}
\sum_{\lambda\in\Lambda_{n,m}}
f(\lambda)\, r_\lambda,
\end{equation}
and every admissible lossless coding scheme with pure reconstructed state, whose local encoders and decoders act as isometries on the relevant supports, satisfies:
\begin{align}
\overline E_{\mathrm G}
\le
\frac{1}{S(n,m)}
\sum_{\lambda\in\Lambda_{n,m}}
f(\lambda)\,
\min\!\left\{
\left\lfloor \frac{m-1}{m}\, r_\lambda \right\rfloor,\,
\log_2\!\Bigl(\prod_{i=1}^{m-1}(\lambda_i+1)\Bigr)
\right\}.
\end{align}
\end{theorem}

\begin{theorem}[Lower bound for $m\geq 3$]
\label{thm:lbm}
Let $t=\lfloor n/2\rfloor$, and let $E_G^{(D)}(\lambda)$ denote the geometric
entanglement of the Dicke state $\ket{D_t^{(n)}}$ across any partition with
block sizes $\lambda=(\lambda_1,\dots,\lambda_m)\in\Lambda_{n,m}$.
If the average rate satisfies
\begin{equation}
R
\ge
\frac{1}{S(n,m)}
\sum_{\lambda\in\Lambda_{n,m}}
f(\lambda)
\sum_{i=1}^m
\left\lceil \log_2(\lambda_i+1)\right\rceil,
\end{equation}
then the average entanglement is lower bounded by
\begin{equation}
\overline E_{\mathrm G}
\ge
\frac{1}{S(n,m)}
\sum_{\lambda\in\Lambda_{n,m}}
f(\lambda)\, E_G^{(D)}(\lambda).
\end{equation}
\end{theorem}

\section{Simulation Study}\label{sec:simulations}
In this section we present the results of our simulation studies. 
For each $n$, we optimize over the complex amplitudes $\mathcal{A}$ and weights $\mathcal{K}$ and numerically solve \eqref{eq:final_opt}. In all simulations, 
we limit the number of Hamming weights in the weighted Dicke states to be $K\leq 3$. We did not experience any improvements in the results by increasing $K$ beyond $3$.
We compare the resulting average entanglement to the corresponding lower and upper bounds obtained in Section~\ref{sec:bounds} for a given average rate. We also return to the motivating example introduced in Section~\ref{sec:motivation}, and provide a simulation study that takes circuit noise and channel noise into account. 

\subsection{Two parties: $m=2$}
For bipartite partitions ($m=2$), the maximal possible values (upper bound) of the GME and the von Neumann entanglement entropy coincide and are given by $\log_2(r)$, i.e., the logarithm of the Schmidt rank $r$ across the bipartitions.    
However, the actual values of the two entanglement measures generally differ. For a bipartite pure state with Schmidt coefficients ${\lambda_i}$, the GME is $E_G = -\log_2 (\lambda_{\max})$,    
where $\lambda_{\max} = \max_i \lambda_i$, whereas the von Neumann entanglement entropy is
$E_{\mathrm{vN}}
=-\sum_i \lambda_i \log_2 (\lambda_i)$. This implies that $E_G \le E_{\mathrm{vN}}$ with equality only when all nonzero Schmidt coefficients are equal, for example for maximally entangled states. Thus, for $m=2$, the von Neumann entropy depends on the entire spectrum, whereas the GME depends only on the largest Schmidt coefficient, which makes the von Neumann entropy more sensitive to changes in the distribution of Schmidt coefficients. Moreover, due to its closed-form, it is computationally simpler to optimize for the von Neumann entropy than the GME. 

Fig.~\ref{fig:m2} presents the performance of optimized weighted Dicke states for the two different optimization criteria. The first set of states is optimized to maximize the von Neumann entropy, and the resulting performance is quantified by the von Neumann entropy. The second set is optimized to maximize the GME, with performance evaluated using the GME.
%
 %
 %
\begin{figure}
    \centering
    \includegraphics[width=14cm]{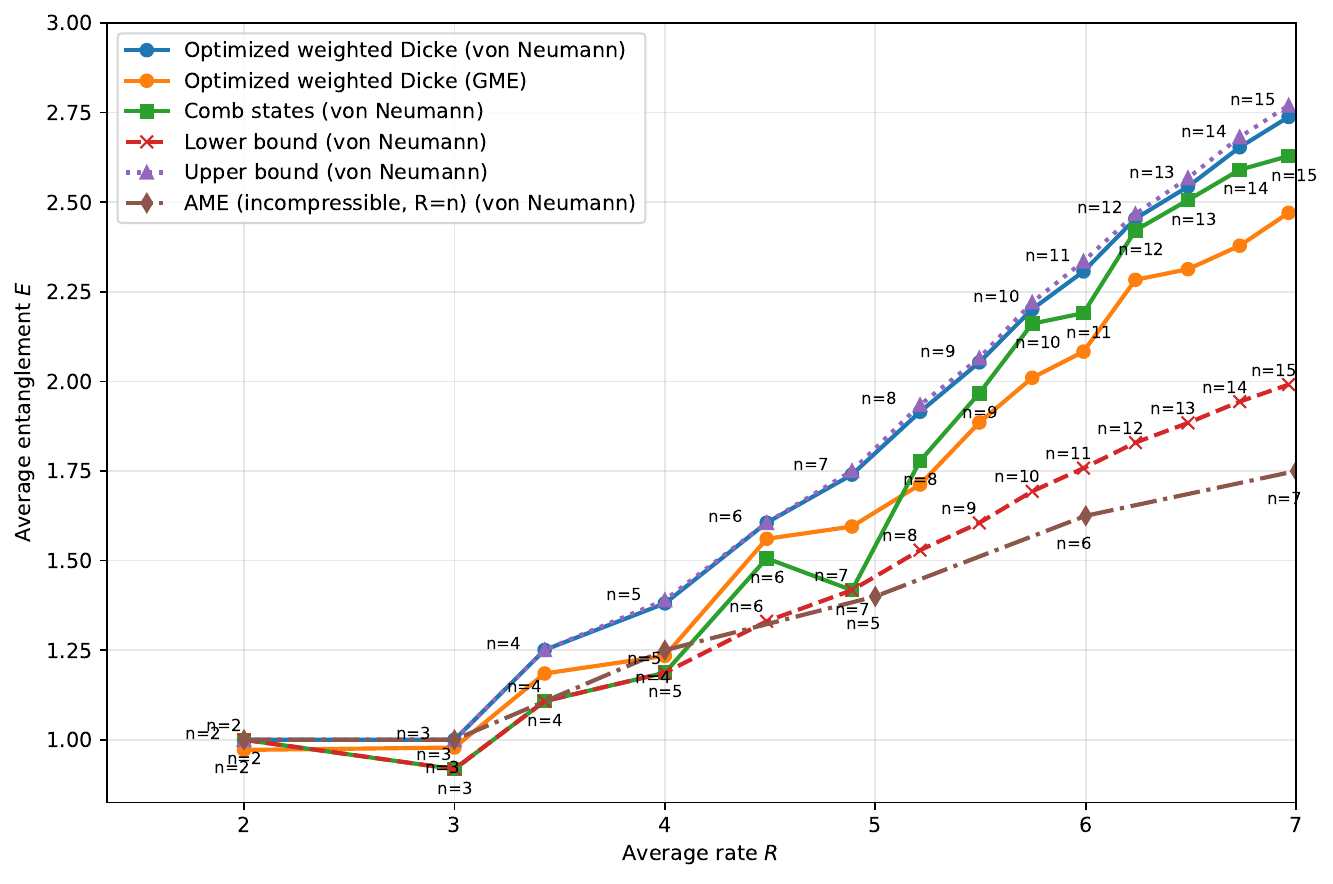}
 \caption{Entanglement-rate curves for $m=2$ and $n=2,\dotsc, 15$. The upper and lower bounds are given by Theorems~\ref{thm:tradeoff_m2} and~\ref{thm:low_m2}, respectively. The weighted Dicke states are the solutions to~\eqref{eq:final_opt} and the Comb states are given by \eqref{eq:comb}. Also shown is the performance of AME states \eqref{eq:ame}. The numbers on the curves indicate $n$.}
    \label{fig:m2}
\end{figure}
 In Fig.~\ref{fig:m2}, it can be observed that for $n=2,3, 4$, and $6$,  the states and compression schemes optimized for the von Neumann entropy have performance reaching the upper bound given by Theorem \ref{thm:tradeoff_m2}, hence 
 they are globally optimal. For all other $n$, the von Neumann optimized states are close to the upper bound. Note that it is in fact unclear whether the upper bound is even achievable in these situations. The curves are not always monotonic or convex. This is partly due to that the states  are not necessarily optimal but also due to that $\log_2(k+1)$ is not always an integer, which means it is rounded upwards to nearest integer (except for the upper bound). In principle, one can apply standard time sharing arguments to convexify the curves. 

Also shown in the figure are the performance for the Comb states. Interestingly,  for $n\geq 10$, the Comb states appears to be close to the upper bound. In \cite{Stockton2003}, it was shown that the Comb states can approach maximal entropy for the fully balanced bipartite cut. However, it can easily be shown that they do not have maximal entanglement over all cuts. Nevertheless, if asymptotically in $n$, the most frequently occurring cuts concentrate near the balanced cut, then the Comb states could potentially be asymptotically optimal. We have not been able to show whether this is satisfied. 

It is interesting to consider the hypothetical case that there would exist AME states for all $n$ and $m=2$. Moreover, 
such states have maximally mixed reductions on the smaller side of every bipartition
and are therefore not losslessly compressible without sacrificing entanglement.
So let us assume that we do not accept a loss in entanglement, in that case the entanglement-rate function of such AME states is given by:
\begin{equation}\label{eq:ame}
\overline{E}_{\mathrm{AME}}(n)
=
\frac{1}{S(n,2)}
\sum_{\lambda \in \Lambda_{n,2}}
f(\lambda)\, \min(\lambda_1,\lambda_2).
\end{equation}
We have plotted this function in Fig.~\ref{fig:m2}. For comparison, we have also plotted the lower bound given in Theorem~\ref{thm:low_m2}. For $n> 6$, the lower bound is greater than $\overline{E}_{\mathrm{AME}}(n)$, which demonstrates the efficiency of compression.

\subsection{Three parties: $m=3$}
In this simulation, we fix $m=3$ and consider systems of $n=3,\dotsc, 15$ qubits. 
The optimal states found for $n=3,\dotsc,10$, are shown in Table~\ref{tab:psi_weights_amplitudes_by_n}, where both the optimal Hamming weights and complex amplitudes are shown. The corresponding GME for each $m$-cut partition is also shown in Table~\ref{tab:Eg_n5_m3}. 
The number of possible $m$-cuts $\{\lambda\}$ is polynomial in $n$ due to the invariance of the GME to symmetry within the states. We have also shown $f(\lambda)$, i.e., the number of occurrences of each $m$-cut $\lambda$. Given $E_G$ and $f(\lambda)$ for each $\lambda$ makes it possible to compute the average entanglement, which is Shown in 
Fig.~\ref{fig:m3}.

For comparison, we also show the resulting partition GME and average entanglement for the Dicke state.
Note that $\ket{W_n} = \ket{D_1^{(n)}},$ i.e., the Dicke state with one excitation.
 For $n=3, m=3$, it is known that $\ket{W_3}$ maximizes the GME, and this is given by $E(W_3) = -\log_2(4/9) \approx 1.1699$. In Table~\ref{tab:Eg_n5_m3}, it can be seen that the numerical optimization is very close to this optimal value. For $n>3$, there are several $m=3$ cuts and one needs to compute GME for each such cuts, which is illustrated in Table~\ref{tab:Eg_n5_m3}. We have also shown $E(\ket{D_n^{(n/2)}})$ in the table, which is the maximal balanced Dicke state, and which is known to have high entanglement \cite{Stockton2003}.
 %
 %

\begin{figure}
    \centering
    \includegraphics[width=14cm]{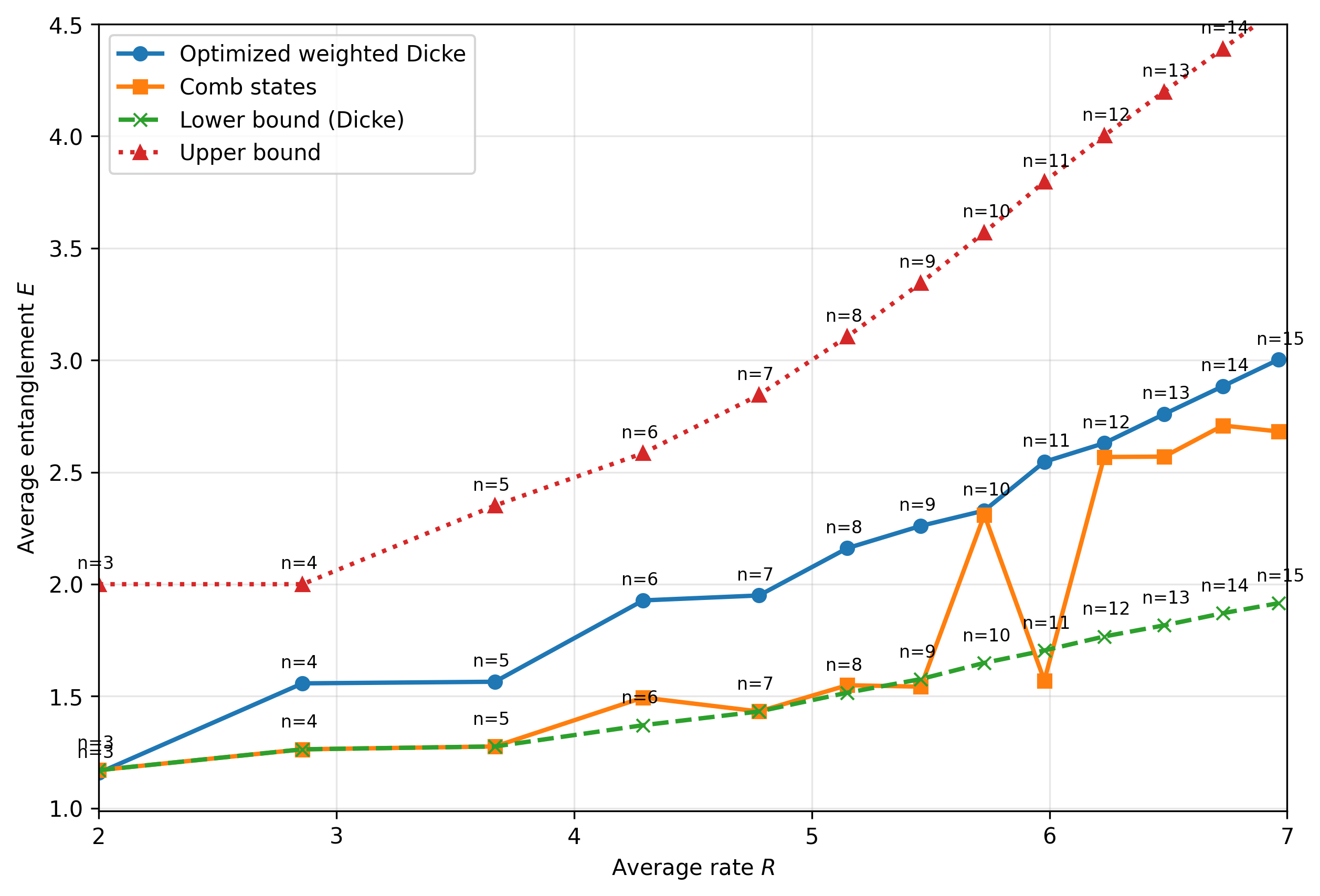}
    \caption{Entanglement-rate curves for $m=3$ and $n=3,\dotsc, 15$. The upper and lower bounds are given by Theorems~\ref{thm:ubm} and~\ref{thm:lbm}, respectively. The weighted Dicke states are the solutions to~\eqref{eq:final_opt} and the Comb states are given by \eqref{eq:comb}. The numbers on the curves indicate $n$.}
    \label{fig:m3}
\end{figure}

\begin{table}[t]
\centering
\scriptsize
\caption{Numerically obtained Hamming-weight support and amplitudes of the optimized symmetric state $\ket{\psi}$ for each $n$ (fixed $m=3$, $K\leq 3$). For a given $n$, the same state is used across all partitions $\lambda$.}
\label{tab:psi_weights_amplitudes_by_n}
\setlength{\tabcolsep}{4pt}
\begin{tabular}{c c c}
\toprule
$n$ & Weights & Amplitudes \\
\midrule
3  & $[0,1,3]$ &
$\begin{array}{@{}c@{}}\left[-0.0463-0.4630i,\; 0.5155-0.2730i,\right.\\ \left.0.2232-0.6272i\right]\end{array}$ \\
\midrule
4  & $[0,3,4]$ &
$\begin{array}{@{}c@{}}\left[0.4640+0.3533i,\; 0.3815+0.7171i,\right.\\ \left.-0.0058-0.0127i\right]\end{array}$ \\
\midrule
5  & $[1,3,5]$ &
$\begin{array}{@{}c@{}}\left[-0.6457-0.3846i,\; -0.0113-0.3494i,\right.\\ \left.-0.5296+0.1801i\right]\end{array}$ \\
\midrule
6  & $[0,3,6]$ &
$\begin{array}{@{}c@{}}\left[-0.2702-0.3852i,\; -0.4450+0.5927i,\right.\\ \left.-0.4371-0.1955i\right]\end{array}$ \\
\midrule
7  & $[1,3,5]$ &
$\begin{array}{@{}c@{}}\left[-0.4760+0.4861i,\; 0.0388+0.0521i,\right.\\ \left.-0.6453+0.3412i\right]\end{array}$ \\
\midrule
8  & $[0,4,8]$ &
$\begin{array}{@{}c@{}}\left[0.4084-0.1626i,\; -0.3501-0.6901i,\right.\\ \left.-0.1451-0.4323i\right]\end{array}$ \\
\midrule
9  & $[1,5,9]$ &
$\begin{array}{@{}c@{}}\left[-0.2159+0.5311i,\; -0.5701+0.3900i,\right.\\ \left.0.1468+0.4155i\right]\end{array}$ \\
\midrule
10 & $[0,3,8]$ &
$\begin{array}{@{}c@{}}\left[-0.1711+0.3480i,\; 0.2868-0.5329i,\right.\\ \left.-0.4506-0.5295i\right]\end{array}$ \\
\bottomrule
\end{tabular}
\end{table}

\begin{table}[t]
\centering
\scriptsize
\caption{Partition based geometric entanglement and coding rates as a function of $n$ for fixed $m=3$ and $K\leq 3$. Three symmetric states are considered. Dicke states $\ket{D^{(n)}_{\lfloor n/2\rfloor}}$ and $\ket{D_1^{(n)}}$  as well as $\ket{\psi}$ obtained numerically as the solution to \eqref{eq:final_opt}.}
\label{tab:Eg_n5_m3}
\setlength{\tabcolsep}{3pt}
\begin{tabular}{@{}c c c c c c c@{}}
\toprule
$n$ & $R$ & $\lambda$ & $E^\lambda_G(\ket{\psi})$ & $E_G^{\lambda}(\ket{D_n^{(1)}})$ & $E_G^{\lambda}(\ket{D_n^{(\lfloor n/2\rfloor)}})$& $f(\lambda)$ \\
\midrule
3  & 3 & $(1,1,1)$ & 1.1599 & 1.1699 & 1.1699 & 1 \\
\midrule
4  & 4 & $(1,1,2)$ & 1.5204 & 1.0000 & 1.2630 & 6 \\
\midrule
5  & 4 & $(1,1,3)$ & 1.4280 & 0.7370 & 1.2775 & 10 \\
   & 5 & $(1,2,2)$ & 1.5933 & 1.1293 & 1.2746 & 15 \\
\midrule
6  & 5 & $(1,1,4)$ & 1.5320 & 0.5850 & 1.3219 & 15 \\
   & 5 & $(1,2,3)$ & 1.8601 & 1.0003 & 1.3949 & 60 \\
   & 6 & $(2,2,2)$ & 1.9593 & 1.1699 & 1.3223 & 15 \\
\midrule
7  & 5 & $(1,1,5)$ & 1.5417 & 0.4854 & 1.3219 & 21 \\
   & 6 & $(1,2,4)$ & 1.8881 & 0.8074 & 1.4035 & 105 \\
   & 5 & $(1,3,3)$ & 2.0098 & 1.0969 & 1.4544 & 70 \\
   & 6 & $(2,2,3)$ & 2.0198 & 1.1293 & 1.4675 & 105 \\
\midrule
8  & 5 & $(1,1,6)$ & 1.5578 & 0.4150 & 1.3479 & 28 \\
   & 6 & $(1,2,5)$ & 1.9744 & 0.6781 & 1.4545 & 168 \\
   & 5 & $(1,3,4)$ & 2.1620 & 1.0004 & 1.5120 & 280 \\
   & 6 & $(2,2,4)$ & 2.1620 & 1.0002 & 1.5071 & 210 \\
   & 6 & $(2,3,3)$ & 2.3268 & 1.1520 & 1.5799 & 280 \\
\midrule
9  & 5 & $(1,1,7)$ & 1.5382 & 0.3626 & 1.3455 & 36 \\
   & 6 & $(1,2,6)$ & 1.9444 & 0.5850 & 1.4595 & 252 \\
   & 6 & $(1,3,5)$ & 2.2020 & 0.8480 & 1.5375 & 504 \\
   & 7 & $(1,4,4)$ & 2.3003 & 1.0768 & 1.5521 & 315 \\
   & 7 & $(2,2,5)$ & 2.2085 & 0.8480 & 1.5567 & 378 \\
   & 7 & $(2,3,4)$ & 2.3337 & 1.1085 & 1.6173 & 1260 \\
   & 6 & $(3,3,3)$ & 2.3599 & 1.1699 & 1.6534 & 280 \\
\midrule
10 & 6 & $(1,1,8)$ & 1.5462 & 0.3219 & 1.3626 & 45 \\
   & 6 & $(1,2,7)$ & 1.9502 & 0.5146 & 1.4884 & 360 \\
   & 6 & $(1,3,6)$ & 2.2518 & 0.7370 & 1.5718 & 840 \\
   & 7 & $(1,4,5)$ & 2.4771 & 1.0000 & 1.6050 & 1260 \\
   & 7 & $(2,2,6)$ & 2.2484 & 0.7370 & 1.5828 & 630 \\
   & 7 & $(2,3,5)$ & 2.4949 & 1.0000 & 1.6688 & 2520 \\
   & 8 & $(2,4,4)$ & 2.5531 & 1.1293 & 1.6760 & 1575 \\
   & 7 & $(3,3,4)$ & 2.5682 & 1.1520 & 1.7153 & 2100 \\
\bottomrule
\end{tabular}
\end{table}

\subsection{Effect of channel and circuit noise}
We now consider the example from Section~\ref{sec:motivation}. In particular, 
we consider the case where a 6-qubit symmetric state is prepared and evenly split into two subsystems each containing three qubits. The subsystems are further "losslessly" compressed from three to two qubits and transmitted to Alice and Bob. We model the circuit noise due to compression and transmission noise separately. We compare the cases with zero and non-zero circuit noise as well as to the case without compression. The results can be seen in Fig.~\ref{fig:entanglement}, where we illustrate the resulting entanglement as a function of channel error probability. 
Note that the inclusion of gate and channel noise renders the states mixed. The GME is well defined for pure states but less interpretable for mixed states. Indeed, the GME increases under depolarizing noise, which makes it unsuitable for our situation. Instead, in this simulation we quantify entanglement using the logarithmic negativity, which was proposed in~\cite{VidalWerner2002}, and later proved in~\cite{Plenio2005LogNegativity} to be an additive entanglement monotone:
\begin{equation}
E_N(\rho_{AB}) = \log_2 \| \rho_{AB}^{T_A} \|_1,
\end{equation}
where $T_A$ denotes the partial transpose with respect to subsystem $A$ and $\|X\|_1$ is the trace norm.
 
In \cite{PivoluskaPlesch2022}, it was shown that one can efficiently implement a transpiled compression circuit that maps 3 qubits to 2 by using 
 9 CNOT gates, 23 RZ rotations, and 14 $\sqrt{X}$ gates. Thus, we assume the compression circuit uses $N_1=37$ single-qubit and $N_2=9$ two-qubit gates.

We use representative superconducting-device parameters consistent with
 IBM Quantum hardware \cite{qiskit-aer}. The coherence times are taken as
$T_1 = 120\,\mu\mathrm{s}$ and $T_2 = 80\,\mu\mathrm{s}$.
Single- and two-qubit gate durations are
$t_{1q} = 35\,\mathrm{ns}$ and $t_{2q} = 300\,\mathrm{ns}$,
with corresponding per-gate depolarizing error probabilities
$p_1 = 10^{-4}$ and $p_2 = 10^{-3}$.
The total effective circuit duration is:
\begin{equation}
t_{\mathrm{tot}} = N_1 t_{1q} + N_2 t_{2q}.
\end{equation}
From the coherence times $T_1$ and $T_2$, we compute the amplitude
damping and phase damping probabilities as
\begin{align}
\gamma_a &= 1 - e^{-t_{\mathrm{tot}}/T_1}, \\
\gamma_p &= 1 - e^{-t_{\mathrm{tot}}/T_\phi},
\qquad
\text{with }
\frac{1}{T_\phi}
=
\frac{1}{T_2}
-
\frac{1}{2T_1}.
\end{align}

The per-gate depolarizing noise is aggregated into effective error
probabilities:
\begin{align}
p_{1,\mathrm{eff}} &= 1 - (1 - p_1)^{N_1}, \\
p_{2,\mathrm{eff}} &= 1 - (1 - p_2)^{N_2}.
\end{align}

The resulting effective single-qubit gate-noise channel is modeled as
\begin{equation}
\mathcal{E}_{\mathrm{gate}}^{(1)}
=
\mathcal{D}(p_{2,\mathrm{eff}})
\circ
\mathcal{D}(p_{1,\mathrm{eff}})
\circ
\mathcal{P}_{\gamma_p}
\circ
\mathcal{A}_{\gamma_a},
\end{equation}
where $\mathcal{A}_{\gamma_a}$ denotes the amplitude-damping channel,
$\mathcal{P}_{\gamma_p}$ the phase-damping channel, and
$\mathcal{D}(p)$ the depolarizing channel.
For the $3|3$ bipartite system, where compression is applied independently
to both three-qubit subsystems $A_3$ and $B_3$, the total gate-noise map is
\begin{equation}
\mathcal{E}_{\mathrm{gate}}^{(6)}
=
\left(
\bigotimes_{q \in A_3} \mathcal{E}_{\mathrm{gate}}^{(1)}
\right)
\otimes
\left(
\bigotimes_{q \in B_3} \mathcal{E}_{\mathrm{gate}}^{(1)}
\right).
\end{equation}

In Fig.~\ref{fig:entanglement}, we present the performance as a function of the channel depolarizing probability $p$. We compare the uncompressed scenario, in which all 6 qubits are transmitted through the channel, with the compressed scenario, where only 4 qubits are transmitted. In addition, we distinguish between ideal compression and compression implemented with noisy gates.
This comparison isolates the fundamental trade-off: compression reduces exposure to channel noise by decreasing the number of transmitted qubits, but simultaneously introduces additional errors due to imperfections in the compression circuit.
As can be seen in Fig.~\ref{fig:entanglement}, there is a cross-over point, where the gain due to compression outweighs the loss due to circuit noise. For this particular setup, the cross-over point happens when the channel error probability is around 0.17. In the absence of circuit noise, the figure shows that there is a clear gain to be achieved since fewer qubits are exposed to the channel noise. These results demonstrate that, under an i.i.d.\ depolarizing noise model, using (noiseless) compression to distribute a fixed amount of entanglement across a smaller number of qubits is advantageous. 

%
\begin{figure}
    \centering
    \includegraphics[width=1.0\linewidth]{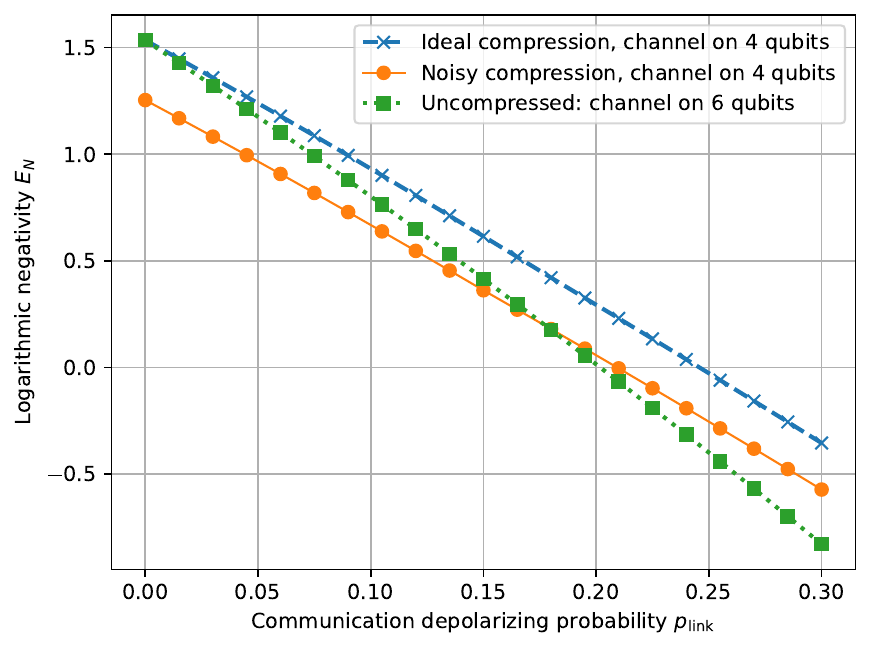}
    \caption{Uncompressed transmission of two three qubit systems with channel noise. Compressed transmission of the systems with and without gate noise followed by channel noise on 4 qubits.}
     \label{fig:entanglement}
\end{figure}

\section{Conclusions}\label{sec:conclusions}

We have shown that permutation-invariant quantum states, combined with symmetric subspace compression, provide an effective framework for distributing multipartite entanglement under resource constraints and uncertain subsystem partitioning. Lossless compression preserves entanglement via local isometries while significantly reducing communication cost, thereby increasing the entanglement per transmitted qubit. Moreover, for suitably structured symmetric states, entanglement can be accessed directly in the compressed code space, eliminating the need for decompression.

\bibliographystyle{ieeetr}
\bibliography{sample_new,bifile}

\newpage
\clearpage   
\pagenumbering{arabic}
\setcounter{page}{1}

\section*{Supplementary Material}
This document contains the proofs of the Theorems in the main paper "Quantum Compression for Distributed Entanglement" by {\O}stergaard et al. 

\begin{IEEEproof}[Proof of Theorem~\ref{thm:tradeoff_m2}]
Fix a bipartition \(\pi = A \mid A^c\), and let
$
k = \min(|A|,|A^c|) \le \left\lfloor \frac{n}{2}\right\rfloor .
$
Since \(\ket{\psi}\) is permutation-invariant, the reduced state on the smaller subsystem is supported on the symmetric subspace of \(k\) qubits, whose dimension is \(k+1\). This implies that:
\begin{equation}
E_{\mathrm{vN}}(\pi)
=
S(\rho_A)
\le
\log_2(k+1),
\label{eq:proof_tradeoff_dimbound}
\end{equation}
where \(\rho_A = \mathrm{Tr}_{A^c}(\ket{\psi}\bra{\psi})\), and we used that the entropy is upper bounded by the logarithm of the support dimension.

For all bipartitions with smaller subsystem size \(k\), we denote the total compression rate by \(r_k=r_A+r_{A^c}\). Since entanglement cannot increase under local encoding and decoding operations, the entanglement after compression and decompression is upper bounded by the entanglement of the compressed state. The latter is a bipartite state supported on registers of dimensions at most \(2^{r_A}\) and \(2^{r_{A^c}}\), respectively. Therefore its Schmidt rank is at most
\begin{equation}
\min\{2^{r_A},2^{r_{A^c}}\}
\le
2^{\lfloor r_k/2 \rfloor},
\end{equation}
which implies that  its bipartite entanglement satisfies
\begin{equation}
E_{\mathrm{vN}}(\pi)
\le
\min\{r_A,r_{A^c}\}
\le
\left\lfloor \frac{r_k}{2}\right\rfloor .
\label{eq:proof_tradeoff_ratebound}
\end{equation}
Combining \eqref{eq:proof_tradeoff_dimbound} and \eqref{eq:proof_tradeoff_ratebound} yields:
\begin{equation}
E_{\mathrm{vN}}(\pi)
\le
\min\!\left\{
\left\lfloor \frac{r_k}{2}\right\rfloor,\,
\log_2(k+1)
\right\}.
\end{equation}

Finally, among all unordered nontrivial bipartitions, exactly \(f(k,n-k)\) have smaller subsystem size \(k\). Averaging the above bound over all such bipartitions gives:
\begin{equation}
\overline{E}_{\mathrm{vN}}
\le
\frac{1}{2^{n-1}-1}
\sum_{k=1}^{\lfloor n/2\rfloor}
f(k,n-k)\,
\min\!\left\{
\left\lfloor \frac{r_k}{2}\right\rfloor,\,
\log_2(k+1)
\right\}.
\end{equation}
This proves the theorem.
\end{IEEEproof}

\begin{IEEEproof}[Proof of Theorem~\ref{thm:low_m2}]
We choose a feasible Dicke state, which is symmetric and let  
 $A|A^c$ be a bipartition with $|A|=k\le n/2$. 
The Dicke state admits the decomposition:
\begin{equation}
\ket{D_t^{(n)}}
=
\sum_{j=0}^{k}
\sqrt{
\frac{\binom{k}{j}\binom{n-k}{t-j}}{\binom{n}{t}}
}
\ket{D_j^{(k)}}\ket{D_{t-j}^{(n-k)}}.
\end{equation}
Hence the reduced state on $A$ is
\begin{equation}
\rho_A
=
\sum_{j=0}^{k}\lambda_j^{(k)}\ket{D_j^{(k)}}\bra{D_k^{(j)}},
\end{equation}
with eigenvalues
\begin{equation}
\lambda_j^{(k)}
=
\frac{\binom{k}{j}\binom{n-k}{t-j}}{\binom{n}{t}}.
\end{equation}
Therefore,
\begin{equation}
E_{\mathrm{vN}}(\pi)=S(\rho_A)=
-\sum_{j=0}^{k}\lambda_j^{(k)}\log_2\lambda_j^{(k)} \triangleq E_k.
\end{equation}
Averaging over all unordered nontrivial bipartitions with $k$ qubits for the smallest side, gives
\begin{equation}
\overline E_{\mathrm{vN}}^{(D)}
=
\frac{1}{2^{n-1}-1}
\sum_{k=1}^{\lfloor n/2\rfloor}f(k,n-k) E_k.
\end{equation}
Since $\ket{D_t^{(n)}}$ is an admissible permutation-invariant pure state, this value is achievable and therefore lower bounds the optimal average entanglement.

For lossless compression, the symmetric subspace of $k$ qubits has dimension $k+1$, so it can be embedded isometrically into $r_k=\lceil\log_2(k+1)\rceil$ qubits. Because local isometries preserve Schmidt coefficients, the von Neumann entanglement is unchanged by the compression. Thus, the achievable coding rate for a cut of size $k$ is the sum of the rates for both $A$ and $A^c$, which gives:
\begin{equation}
\lceil\log_2(k+1)\rceil + \lceil\log_2((n-k)+1)\rceil.
\end{equation}
Averaging over all bipartitions yields the stated rate and completes the proof.
\end{IEEEproof}

\begin{IEEEproof}[Proof of Theorem~\ref{thm:ubm}]
Fix a partition $\pi=\{S_1,\dots,S_m\}$ with block sizes
$\lambda_1\le \cdots \le \lambda_m$, and let
$\ket{\widetilde\psi_\pi}$ denote the compressed pure state on the output
registers of the local encoders. Since the coding scheme is lossless and the
reconstructed state is pure, the local decoding maps act isometrically on the
support of the compressed state and therefore preserve the GME. Hence it suffices to bound the GME of
$\ket{\widetilde\psi_\pi}$.

We use the dimension bound for the geometric entanglement as done in~\cite{WeiGoldbart2003}.
Let the local Hilbert-space dimensions of $\ket{\widetilde\psi_\pi}$ be
$d_1,\dots,d_m$, with the ordering
$d_{(1)}\le d_{(2)}\le \cdots \le d_{(m)}$.
Consider the bipartition
$S_{(1)} \,\big|\, S_{(2)}\cdots S_{(m)}$.
The Schmidt rank of $\ket{\widetilde\psi_\pi}$ across this cut is at most $d_{(1)}$.
Hence the largest Schmidt coefficient squared is at least $1/d_{(1)}$, so there exist normalized vectors
$\ket{\alpha_{(1)}}\in\mathcal H_{S_{(1)}}$ and
$\ket{\eta_{(2\cdots m)}}\in\bigotimes_{j=2}^m \mathcal H_{S_{(j)}}$
such that
\begin{equation}
\left|
\langle \alpha_{(1)}\otimes \eta_{(2\cdots m)} \mid \widetilde\psi_\pi\rangle
\right|^2
\ge \frac{1}{d_{(1)}}.
\end{equation}
We now proceed by regarding $\ket{\eta_{(2\cdots m)}}$ as a pure state on the remaining
$m-1$ parties and bipartition it as
$S_{(2)} \,\big|\, S_{(3)}\cdots S_{(m)}$.
Its Schmidt rank is at most $d_{(2)}$, so there exist normalized vectors
$\ket{\alpha_{(2)}}\in\mathcal H_{S_{(2)}}$ and
$\ket{\eta_{(3\cdots m)}}\in\bigotimes_{j=3}^m \mathcal H_{S_{(j)}}$
such that
\begin{equation}
\left|
\langle \alpha_{(2)}\otimes \eta_{(3\cdots m)} \mid \eta_{(2\cdots m)}\rangle
\right|^2
\ge \frac{1}{d_{(2)}}.
\end{equation}
Iterating this argument through the first $m-1$ local systems yields normalized
product vectors $\ket{\alpha_{(1)}},\dots,\ket{\alpha_{(m)}}$ satisfying
\begin{equation}
\left|
\langle \alpha_{(1)}\otimes\cdots\otimes\alpha_{(m)} \mid \widetilde\psi_\pi\rangle
\right|^2
\ge
\frac{1}{\prod_{i=1}^{m-1} d_{(i)}}.
\end{equation}
Therefore,
\begin{equation}
E_{\mathrm G}(\pi)
\le
\log_2\!\Bigl(\prod_{i=1}^{m-1} d_{(i)}\Bigr).
\label{eq:EG_general_dim_bound}
\end{equation}

Let the $m$ encoders use $r_1,\dots,r_m$ qubits, respectively.
Then the encoded local dimensions satisfy $d_i\le 2^{r_i}$. Hence
\begin{equation}
E_{\mathrm G}(\pi)
\le
\sum_{i=1}^{m-1} r_{(i)},
\label{eq:EG_general_rate_pre}
\end{equation}
where $r_{(1)}\le \cdots \le r_{(m)}$ are the ordered rates.
Since
\begin{equation}
\sum_{i=1}^{m-1} r_{(i)}
=
r(\pi)-r_{(m)},
\end{equation}
and since the largest local rate must be greater than or equal to the average, we have that:
\begin{equation}
r_{(m)}\ge \left\lceil \frac{r(\pi)}{m}\right\rceil,    
\end{equation}
which implies that
\begin{equation}
\sum_{i=1}^{m-1} r_{(i)}
\le
r(\pi)-\left\lceil \frac{r(\pi)}{m}\right\rceil
=
\left\lfloor \frac{m-1}{m}\,r(\pi)\right\rfloor.
\end{equation}
Combining this with~\eqref{eq:EG_general_rate_pre} gives
\begin{equation}
E_{\mathrm G}(\pi)
\le
\left\lfloor \frac{m-1}{m}\,r(\pi)\right\rfloor.
\label{eq:EG_general_rate_bound}
\end{equation}

On the other hand, because the source state is permutation-invariant, each block
of size $\lambda_i$ is supported on the symmetric subspace of dimension
$\lambda_i+1$. Hence the local support dimensions satisfy
$d_i\le \lambda_i+1, i=1,\dots,m$.
Substituting this into~\eqref{eq:EG_general_dim_bound} yields
\begin{equation}
E_{\mathrm G}(\pi)
\le
\log_2\!\Bigl(\prod_{i=1}^{m-1}(\lambda_i+1)\Bigr),
\label{eq:EG_general_sym_bound}
\end{equation}
since $\lambda_1\le \cdots \le \lambda_m$ and therefore the product of the
$m-1$ smallest local dimensions is
$\prod_{i=1}^{m-1}(\lambda_i+1)$.

Combining~\eqref{eq:EG_general_rate_bound} and
\eqref{eq:EG_general_sym_bound} gives
\begin{equation}
E_{\mathrm G}(\pi)
\le
\min\!\left\{
\left\lfloor \frac{m-1}{m}\,r(\pi)\right\rfloor,\,
\log_2\!\Bigl(\prod_{i=1}^{m-1}(\lambda_i+1)\Bigr)
\right\}.
\end{equation}

The bound depends only on the size pattern $\lambda$, so averaging over all
unordered $m$-partitions, of which $f(\lambda)$ have block sizes $\lambda$,
yields
\begin{equation}
\overline E_{\mathrm G}
\le
\frac{1}{S(n,m)}
\sum_{\lambda\in\Lambda_{n,m}}
f(\lambda)\,
\min\!\left\{
\left\lfloor \frac{m-1}{m}\,r_\lambda\right\rfloor,\,
\log_2\!\Bigl(\prod_{i=1}^{m-1}(\lambda_i+1)\Bigr)
\right\}.
\end{equation}
This proves the theorem.
\end{IEEEproof}

\begin{IEEEproof}[Proof of Theorem~\ref{thm:lbm}]
We prove the bound by constructing a feasible source state and a feasible
lossless compression scheme. Choose the permutation-invariant pure state
$\ket{\psi}=\ket{D_t^{(n)}},
\; t=\lfloor n/2\rfloor$,
and fix an $m$-partition $\pi=\{S_1,\dots,S_m\}$
with block sizes
\begin{equation}
|S_i|=\lambda_i,\quad
1\le \lambda_1\le \cdots \le \lambda_m,\quad
\sum_{i=1}^m \lambda_i=n.    
\end{equation}
The Dicke state admits the multipartite decomposition~\cite{Moreno2018,Stockton2003}:
\begin{equation}
\ket{D_t^{(n)}}
=
\sum_{\substack{j_1,\dots,j_m\ge 0\\ j_i\le \lambda_i,\ \sum_{i=1}^m j_i=t}}
\sqrt{
\frac{\prod_{i=1}^m \binom{\lambda_i}{j_i}}{\binom{n}{t}}
}
\;
\ket{D_{\lambda_1}^{(j_1)}}\otimes\cdots\otimes\ket{D_{\lambda_m}^{(j_m)}}.
\label{eq:general_dicke_decomp}
\end{equation}
Thus, $\ket{D_t^{(n)}}$ lies in $\bigotimes_{i=1}^m \mathrm{Sym}^{\lambda_i}(\mathbb C^2)$.
Therefore, when computing the geometric entanglement across the partition
$\pi$, the optimization may be restricted to product vectors in these local
symmetric subspaces. For the $i$th local symmetric subspace we define the state $\ket{u_i}$ by expanding onto the Dicke basis:
\begin{equation}
\ket{u_i}
=
\sum_{j_i=0}^{\lambda_i}
u_{i,j_i}\ket{D_{\lambda_i}^{(j_i)}},
\qquad
\|\ket{u_i}\|=1,
\end{equation}
where $u_{i,j_i} \in \mathbb{C}$ are the expansion coefficients of $\ket{u_i}$ in the Dicke basis, satisfying $\sum_{j_i=0}^{\lambda_i} |u_{i,j_i}|^2 = 1$.
The maximal product overlap is then given by:
\begin{equation}
\Lambda_\lambda
=
\max_{\|u_1\|=\cdots=\|u_m\|=1}
\left|
\sum_{\substack{j_1,\dots,j_m\ge 0\\ j_i\le \lambda_i,\ \sum_{i=1}^m j_i=t}}
\sqrt{
\frac{\prod_{i=1}^m \binom{\lambda_i}{j_i}}{\binom{n}{t}}
}
\;
\prod_{i=1}^m u_{i,j_i}
\right|.
\end{equation}

Hence the geometric entanglement of the Dicke state across any partition with
size pattern $\lambda$ is
\begin{equation}
E_G^{(D)}(\lambda)
=
-\log_2 \Lambda_\lambda^2.
\end{equation}
By permutation invariance, this value depends only on the block-size multiset
$\lambda$ and not on the specific realization of the partition.

We now compress each block losslessly into its symmetric subspace.
Since
$\dim \mathrm{Sym}^{\lambda_i}(\mathbb C^2)=\lambda_i+1$,
the $i$th block can be represented using
$\left\lceil \log_2(\lambda_i+1)\right\rceil$
qubits via a local isometric embedding.
Thus, for the size pattern $\lambda$, the total rate is
\begin{equation}
r_\lambda
=
\sum_{i=1}^m
\left\lceil \log_2(\lambda_i+1)\right\rceil.
\end{equation}

Because the coding is lossless, the encoder and decoder act as local isometries
on the local symmetric subspaces. Local isometries preserve the geometric
entanglement of pure states, so the reconstructed state has the same geometric
entanglement as the original Dicke state across the partition.
Therefore, for every partition of size pattern $\lambda$, the achievable
entanglement equals $E_G^{(D)}(\lambda)$.

Averaging over all unordered $m$-partitions, of which $f(\lambda)$ have
block sizes $\lambda$, gives the achievable average entanglement
\begin{equation}
\overline E_{\mathrm G}^{(D)}
=
\frac{1}{S(n,m)}
\sum_{\lambda\in\Lambda_{n,m}}
f(\lambda)\,E_G^{(D)}(\lambda).
\end{equation}
The corresponding average rate is
\begin{equation}
R^{(D)}
=
\frac{1}{S(n,m)}
\sum_{\lambda\in\Lambda_{n,m}}
f(\lambda)
\sum_{i=1}^m
\left\lceil \log_2(\lambda_i+1)\right\rceil.
\end{equation}
Hence, whenever the allowed average rate satisfies
$
R\ge R^{(D)},
$
this Dicke-state construction is feasible and achieves
\begin{equation}
\overline E_{\mathrm G}
\ge
\overline E_{\mathrm G}^{(D)}
=
\frac{1}{S(n,m)}
\sum_{\lambda\in\Lambda_{n,m}}
f(\lambda)\,E_G^{(D)}(\lambda).    
\end{equation}
This proves the theorem.
\end{IEEEproof}

\end{document}